\documentclass[a4paper,UKenglish,cleveref, autoref, thm-restate]{lipics-v2021}

\pdfoutput=1 
\hideLIPIcs  

\title{Constant-Depth Sorting Networks} 

\author{Natalia Dobrokhotova-Maikova}{Yandex, Moscow, Russia}{dobromayk@yandex.ru
}{}{}

\author{Alexander Kozachinskiy}{Steklov Mathematical Institute of Russian Academy of Sciences, Moscow, Russia}{kozlach@mail.ru}{https://orcid.org/0000-0002-9956-9023}{The work of A.~Kozachinskiy was performed at the Steklov International Mathematical Center and supported by the Ministry of Science and Higher Education of the Russian Federation (agreement no. 075-15-2022-265).
}

\author{Vladimir Podolskii}{Steklov Mathematical Institute of Russian Academy of Sciences, Moscow, Russia \and HSE University, Moscow, Russia}{podolskii.vv@gmail.com}{https://orcid.org/0000-0001-7154-138X}{}

\authorrunning{N. Dobrokhotova-Maikova, A. Kozachinskiy and V. Podolskii} 

\Copyright{Natalia Dobrokhotova-Maikova, Alexander Kozachinskiy and Vladimir Podolskii} 



\category{} 

\relatedversion{} 

\acknowledgements{}

\nolinenumbers 

\EventEditors{John Q. Open and Joan R. Access}
\EventNoEds{2}
\EventLongTitle{42nd Conference on Very Important Topics (CVIT 2016)}
\EventShortTitle{CVIT 2016}
\EventAcronym{CVIT}
\EventYear{2016}
\EventDate{December 24--27, 2016}
\EventLocation{Little Whinging, United Kingdom}
\EventLogo{}
\SeriesVolume{42}
\ArticleNo{23}

\usepackage{tikz}
\usetikzlibrary{shapes.geometric,patterns}
\newcommand{\cube}[3][fill=white]{
    \draw[black, #1] (#2) -- ++(-#3,0,0) -- ++(0,#3,0) -- ++(0,0,-#3) -- ++(#3,0,0) -- ++(0,-#3,0)-- cycle;
    \draw[black,line join=bevel, #1] (#2) -- ++(0,#3,0) -- ++(0,0,-#3) (#2)++(0,#3,0)--++(-#3,0,0);
}

\colorlet{one}{black!60!white}
\colorlet{two}{black!20!white}

\begin{document}

\maketitle

\begin{abstract}
In this paper, we address sorting networks that are constructed from comparators of arity $k > 2$. That is, in our setting the arity of the comparators --- or, in other words, the number of inputs that can be sorted at the unit cost ---  is a parameter. We study its relationship with two other parameters --- $n$, the number of inputs, and $d$, the depth.

This model received considerable attention. Partly, its motivation is to better understand the structure of sorting networks. In particular, sorting networks with large arity are related to recursive constructions of ordinary sorting networks.
Additionally, studies of this model have natural correspondence with  a recent line of work on constructing circuits for majority functions from majority gates of lower fan-in.

Motivated by these questions, we obtain the first lower bounds on the arity of constant-depth sorting networks.
More precisely, we consider sorting networks of depth $d$ up to 4, and determine the minimal $k$ for which there is such a network with comparators of arity $k$. For depths $d=1,2$ we observe that $k=n$. For $d=3$ we show that $k = \lceil \frac n2 \rceil$. For $d=4$ the minimal arity becomes sublinear: $k = \Theta(n^{2/3})$. This contrasts with the case of majority circuits, in which $k = O(n^{2/3})$ is achievable already for depth $d=3$. 
\end{abstract}

\section{Introduction}

A sorting network receives an array of numbers and outputs the same numbers in the non-decreasing  order.
It consists of \emph{comparators} that can swap any two numbers from the array if they are not in the non-decreasing order. The main parameters of a sorting network are the size, that is, the number of comparators, and the depth, that is, the number of layers in the network, where each layer consists of several comparators applied to disjoint pairs of variables. Sorting networks are a classical model in theoretical computer science with vast literature devoted to them, see, for example~\cite{Batcher68,AjtaiKS83,Leighton85,
paterson1990improved,Parberry92,KahaleLMPSS95,seiferas2009sorting,BundalaZ14}, see also the Knuth's book~\cite{Knuth98} and the Baddar's and Batcher's book~\cite{BaddarB12}. Despite considerable efforts, still there are many open problems related to sorting networks.

In this paper, our main interest is the depth of  sorting networks. 
There is a folklore
$(2 - O(1))\log_2 n$  depth lower bound for networks sorting $n$ numbers. 
It was improved by Yao~\cite{Yao80} and later by Kahale et al.~\cite{KahaleLMPSS95} with the current record about $3.27 \log_2 n$. As for  upper bounds, a construction with $O(\log n)$ depth was given in~\cite{AjtaiKS83} and is usually referred to as the AKS sorting network. However, this construction is very complicated and impractical due to a large constant hidden in the O-notation.
 There are some simplifications and improvements of this construction~\cite{paterson1990improved,seiferas2009sorting}, but they did not make it practical.
On the other hand, there are several simple and practical constructions
of depth $\Theta(\log^2 n)$~\cite{Knuth73,Batcher68,Parberry92}.

One could also consider sorting networks with comparators that have $k > 2$ inputs. We will call them $k$-sorting networks.
They appear in the literature 
since 70s, the setting is mentioned already in the Knuth's book~\cite[Problem 5.3.4.54]{Knuth98}, followed by numerous works~\cite{TsengL85,ParkerP89,BeigelG90,NakataniHAT89,CypherS92,LeeB95,ShiYW14,gao1997sloping,zhao1998efficient}.
They are usually studied to better understand the structure of ordinary sorting networks. More specifically, they are closely related to recursive constructions of sorting networks. Having a good construction of a $k$-sorting network, one can apply it to its own comparators, getting a construction with smaller $k$, until eventually $k$ becomes 2, and we get an ordinary sorting network.

Any $k$-sorting network with $n$ inputs must have depth at least $\log_k n$ because otherwise outputs cannot be connected to all $n$ inputs.
Parker and Parbery~\cite{ParkerP89} constructed a simple and potentially practical $k$-sorting network of depth $\leq 4 \log^2_k n$ (in case when $n$ is an integral power of $k$).
 At the same time, as  Chv{\'a}tal shows in his lecture notes~\cite{Chvatal92}, 
the AKS sorting network also generalizes to this setting, giving a construction of depth $O(\log_k n)$. However, as with the AKS sorting network itself, this construction is complicated and impractical. 
So the search for simple constructions continues.

There is a related setting of computing 
majority function by monotone Boolean circuits.  Majority function receives as input a sequence of $n$ bits and outputs 1 if and only if more than a half of the inputs are 1's. Monotone Boolean circuits consist of AND and OR gates of fan-in 2. There is a quest of constructing a monotone Boolean circuit for majority which is simple and practical and  has depth $O(\log n)$. Note that it  can only be easier than an analogous quest for sorting networks. This is because  a sorting network can be transformed into a monotone Boolean circuit which  computes majority and has the same depth.
Indeed, if we restrict inputs to $\{0,1\}^n$, then each comparator can be simulated by a pair of AND and OR gates (AND computes the minimum of two Boolean inputs and OR computes the maximum). And the majority is just the median bit of the sorted array. In particular, we get an  $O(\log n)$-depth monotone circuit for majority from the AKS sorting network. Yet again, the resulting circuit has the same disadvantages as the AKS construction. But in contrast to sorting networks, there is an alternative construction of a monotone depth-$O(\log n)$ Boolean circuit for majority due to Valiant~\cite{Valiant84}. His construction is simple and has
 a reasonable constant hidden in the O-notation, but it is randomized. It was partially derandomized and made closer to practice by Hoory, Magen and Pitassi~\cite{hoory2006monotone}. But still all known fully deterministic constructions that are simple and practical  are of depth $\Theta(\log^2 n)$.

Just as with sorting networks, we can consider monotone circuits with 
majority gates 
of fan-in at most $k$. Similarly to the case $k=2$, one can use $k$-sorting networks for constructing such circuits.
Recently, considerable attention was given to the following question: given $d$ and $n$, what is the minimal $k$ for which there is a monotone circuit with majority gates of fan-in $k$,
 computing majority on $n$ bits?
This was first studied by Kulikov and Podolskii~\cite{KulikovP19}. 
There are two natural models that can be considered here. In the first one,
 we allow to use only majority functions as gates of the circuit. Analogously to fan-in 2 model, it is not harder to construct such a circuit than a sorting network. In the second one, 
 we allow to use arbitrary threshold functions as gates. A threshold function with threshold $t$ outputs 1 iff at least $t$ of its inputs are equal to 1. Clearly, the first model is more restrictive. However, the power of the models is the same up to a constant factor: one can just simulate a
 threshold gate as a majority gate adding constant inputs.

For depth-2 circuits, it was shown in~\cite{KulikovP19} that $k = \Omega\left(n^{13/19}\cdot (\log n)^{-2/19}\right)$. For the case when the weights of inputs on the first layer of the circuit are equal to one, this lower bound was improved in the series of works~\cite{EngelsGMR20,HrubesRRY19} to $k \geq n/2 -o(n)$. For the general $d$ it was shown in~\cite{KulikovP19} that
\[
k = \Omega\left(\frac{n^{14/(7d+6)}}{(\log n)^{4/(7d+6)} }\right).
\]

As for the upper bounds, for depth $d=2$ it was shown that $k \leq n-2$ for the model with standard majority circuits~\cite{kombarov2017,amano2017} and $k \leq \frac{2n}{3}+O(1)$ for the model with gates having arbitrary thresholds~\cite{bruno,Posobin17}.
For depth $d=3$,
 there is a circuit with $k = O (n^{2/3})$~\cite{KulikovP19}.

 Finally, Hrubes and Rao~\cite{HrubesR15} studied a more general problem of computing Boolean functions by depth-2 circuits, consisting of arbitrary functions of small fan-in as gates. Recently, Lecomte, Ramakrishnan and Tan~\cite{LecomteRT22} obtained tight lower bounds for majority function in this model.

\paragraph*{Our results}
We initiate the studies of constant-depth $k$-sorting networks. As with the case of majority circuits, we fix  $d$ and try to find the minimal $k$ for which a  $k$-sorting network of depth $d$ exists. Namely, we determine the minimal possible value of $k$ for sorting networks of depth up to 4. For depth $d=1,2$ we show that $k$ must be equal to $n$ (the number of inputs), so that no non-trivial network exists in this case. For depth $d=3$ we show that the optimal $k$ is equal to $\lceil n/2 \rceil$. For depth $d=4$ we prove a lower bound $k = \Omega(n^{2/3})$. A matching upper bound is implicit in~\cite[Section 2]{Leighton85}. For arbitrary depth, we show that $k = \Omega\left(n^{1/\lceil \frac d2 \rceil}\right)$. When written as a lower bound on $d$, the last bound gives us $d\ge 2\log_k n - O(1)$.  It is twice as strong as a trivial lower bound $d\ge \log_k n$, mentioned previously.

We note that we obtain larger values of $k$ for sorting networks than for majority circuits of the same depth.
 For example, $k$ becomes sublinear only for depth 4 in the sorting network setting, whereas for majority circuits we get the same upper bound $k=O(n^{2/3})$ for $d=3$. 
This can be viewed as another indication that 
constructing sorting networks is strictly harder in general than constructing circuits for majority. On one hand, the fan-in becomes sublinear one step earlier for majority circuits than for sorting networks.
On the other hand, we have methods that prove matching lower bounds for sorting networks of depth up to 4, whereas for majority circuits 
there are no
strong lower bounds even for depth $3$.

\medskip 

The rest of the paper is organized as follows. In Section~\ref{sec:techniques} we provide the necessary preliminary information, introduce the main techniques and use it to obtain the first results -- a bound for arbitrary depth and bounds for depth 1 and 2. In Section~\ref{sec:3} we prove the results on sorting networks of depth 3. In Section~\ref{sec:4} we prove the result on sorting networks of depth~4.

\section{The Techniques and Initial Results}
\label{sec:techniques}

We start with some terminology related to sorting networks.
A depth-$d$ network with $n$ inputs consists of $d + 1$ arrays enumerated from $0$ to $d$, each having $n$ \emph{cells}. We imagine these arrays arranged from left to right, see Figure~\ref{example} for an example.  We denote cells in the arrays by $c_{a,b}$ for $a=0,\ldots, d$ and $b=1,\ldots, n$. Cells $c_{0,b}$ are input cells, each of these cells can contain an arbitrary integral number. We will refer to cell $c_{d, b}$'s as output cells of the network.

Next, between any two consecutive arrays of cells there is a \emph{layer of comparators}. They are enumerated from 1 to $d$, so the $a$th layer is between the $(a-1)$st array and  the $a$th array. Formally, a layer of comparators is a partition of the set $\{1, 2, \ldots, n\}$ into subsets called \emph{comparators}. The maximal size of a comparator over all layers is called the \emph{arity} of the network (we usually denote the arity by $k$).

If $I\subseteq\{1, 2, \ldots, n\}$ is a comparator from the $a$th layer, then it  is applied to the cells $c_{a-1, i}, i\in I$ (we will refer to these cells as input cells of $I$). It sorts them in the non-decreasing order and writes the results into the cells $c_{a, i}, i\in I$ (we will refer to them as output cells of $I$). In this way, given an assignment of the input cells of the network, one inductively (from left to right) defines the values of the remaining cells of the network. We say that a network is \emph{sorting} if for any input the array with index $d$ (one with the output cells) is sorted.

\begin{figure}[h!]
\centering
\begin{tikzpicture}[square/.style={regular polygon,regular polygon sides=4},scale=0.8,every node/.style={scale=0.8}]

\draw[->, line width=1mm] (-1,0) -- (-1,7);

	\node at (-2,0) [draw=none] (0) {};
\node at (1.5,7.5) [draw=none] (a) {1st layer};
\node at (4.5,7.5) [draw=none] (b) {2nd layer};
\node at (7.5,7.5) [draw=none] (c) {3rd layer};

        \node at (0,0) [square,draw,thick,minimum width=1.3cm] (l11) {$1$};
        \node at (0,1) [square,draw, thick, minimum width=1.3cm] (l12) {$5$};
        \node at (0,2) [square,draw,thick, minimum width=1.3cm] (l13) {$8$};
        \node at (0,3) [square,draw,thick,minimum width=1.3cm] (l14) {$2$};
        \node at (0,4) [square,draw,thick,minimum width=1.3cm] (l15) {$42$};
        \node at (0,5) [square,draw,thick,minimum width=1.3cm] (l16) {$27$};
        \node at (0,6) [square,draw,thick,minimum width=1.3cm] (l17) {$7$};
        \node at (0,7) [square,draw,thick,minimum width=1.3cm] (l18) {$4$};

		 \node at (3,0) [square,draw,thick,minimum width=1.3cm] (l21) {$1$};
        \node at (3,1) [square,draw, thick, minimum width=1.3cm] (l22) {$2$};
        \node at (3,2) [square,draw,thick, minimum width=1.3cm] (l23) {$5$};
        \node at (3,3) [square,draw,thick,minimum width=1.3cm] (l24) {$8$};
        \node at (3,4) [square,draw,thick,minimum width=1.3cm] (l25) {$4$};
        \node at (3,5) [square,draw,thick,minimum width=1.3cm] (l26) {$7$};
        \node at (3,6) [square,draw,thick,minimum width=1.3cm] (l27) {$27$};
        \node at (3,7) [square,draw,thick,minimum width=1.3cm] (l28) {$42$};

\draw [-] (l11) -- (l21) node [midway, draw,circle,fill=one] {};
 \draw [-] (l12) -- (l22) node [midway, draw,circle,fill=one] {};
 \draw [-] (l13) -- (l23) node [midway, draw,circle,fill=one] {};
 \draw [-] (l14) -- (l24) node [midway, draw,circle,fill=one] {};
 \draw [-] (l15) -- (l25) node [midway, draw,circle,fill=two] {};
 \draw [-] (l16) -- (l26) node [midway, draw,circle,fill=two] {};
 \draw [-] (l17) -- (l27) node [midway, draw,circle,fill=two] {};
 \draw [-] (l18) -- (l28) node [midway, draw,circle,fill=two] {};
 
		 \node at (6,0) [square,draw,thick,minimum width=1.3cm] (l31) {$1$};
        \node at (6,1) [square,draw, thick, minimum width=1.3cm] (l32) {$2$};
        \node at (6,2) [square,draw,thick, minimum width=1.3cm] (l33) {$5$};
        \node at (6,3) [square,draw,thick,minimum width=1.3cm] (l34) {$8$};
        \node at (6,4) [square,draw,thick,minimum width=1.3cm] (l35) {$4$};
        \node at (6,5) [square,draw,thick,minimum width=1.3cm] (l36) {$7$};
        \node at (6,6) [square,draw,thick,minimum width=1.3cm] (l37) {$27$};
        \node at (6,7) [square,draw,thick,minimum width=1.3cm] (l38) {$42$};

\draw [-] (l21) -- (l31) node [midway, draw,circle,fill=one] {};
 \draw [-] (l22) -- (l32) node [midway, draw,circle,fill=one] {};
 \draw [-] (l23) -- (l33) node [midway, draw,circle,fill=two] {};
 \draw [-] (l24) -- (l34) node [midway, draw,circle,fill=two] {};
 \draw [-] (l25) -- (l35) node [midway, draw,circle,fill=one] {};
 \draw [-] (l26) -- (l36) node [midway, draw,circle,fill=one] {};
 \draw [-] (l27) -- (l37) node [midway, draw,circle,fill=two] {};
 \draw [-] (l28) -- (l38) node [midway, draw,circle,fill=two] {};

\node at (9,0) [square,draw,thick,minimum width=1.3cm] (l41) {$1$};
        \node at (9,1) [square,draw, thick, minimum width=1.3cm] (l42) {$2$};
        \node at (9,2) [square,draw,thick, minimum width=1.3cm] (l43) {$4$};
        \node at (9,3) [square,draw,thick,minimum width=1.3cm] (l44) {$5$};
        \node at (9,4) [square,draw,thick,minimum width=1.3cm] (l45) {$7$};
        \node at (9,5) [square,draw,thick,minimum width=1.3cm] (l46) {$8$};
        \node at (9,6) [square,draw,thick,minimum width=1.3cm] (l47) {$27$};
        \node at (9,7) [square,draw,thick,minimum width=1.3cm] (l48) {$42$};

 \draw [-] (l31) -- (l41) node [midway, draw,circle,fill=two] {};
 \draw [-] (l32) -- (l42) node [midway, draw,circle,fill=two] {};
 \draw [-] (l33) -- (l43) node [midway, draw,circle,fill=one] {};
 \draw [-] (l34) -- (l44) node [midway, draw,circle,fill=one] {};
 \draw [-] (l35) -- (l45) node [midway, draw,circle,fill=one] {};
 \draw [-] (l36) -- (l46) node [midway, draw,circle,fill=one] {};
 \draw [-] (l37) -- (l47) node [midway, draw,circle,fill=two] {};
 \draw [-] (l38) -- (l48) node [midway, draw,circle,fill=two] {};

    \end{tikzpicture}

  \caption{A sorting network with $n = 8, d = 3, k = 4$. There are 4 arrays of cells going from left to right. The arrow shows the ordering of the arrays (it goes from small to large indices).  The comparators of the network are given by the colors of circles over the lines. For example, the first layer has 2 comparators -- the dark gray one sorts the first 4 input cells, and the light gray one sorts the last 4 input cells.}
\label{example}
\end{figure}
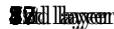

Sometimes, to avoid overusing indices of the arrays, we will refer to cells of the $a$th array as cells \emph{after the $a$th layer} (of comparators) or \emph{before the $(a+1)$st layer}. For example, input cells are cells before the first layer and cells of the $d$th array are cells after the last layer.

We will also say that some cells are of the network are \emph{connected} to each other. To define this, for each comparator we draw a directed edge from every input cell of this comparator to every output cell of this comparator. Then we say that two cells are connected if there is a directed path between them in the resulting graph.

In our upper bound for $d = 3$ we will employ the following classical principle.
\begin{lemma}[Zero-one principle~\cite{Knuth98}]
A network with $n$ inputs sorts all integer sequences in the non-decreasing order if and only if it sorts all sequences from $\{0,1\}^n$ in the non-decreasing order.
\end{lemma}
By this principle, when constructing sorting networks, we can assume that each input cell receives either $0$ or $1$. In fact, we consider only binary inputs in our lower bounds as well.

In the rest of this section, we discuss our technique. It relies on the notion of \emph{access}. We say that we have access to a cell $c_{a,b}$ on input $x = (x_1, \dots, x_n) \in \{0,1\}^n$ if its value on $x$ is $0$, but it is possible to change one of the inputs $x_i$ from 0 to 1 in such a way that $c_{a,b}$ will become 1.

For example, before the first layer we have access to all cells that contain 0. If the network sorts all inputs correctly, then after the last layer we have access to a single cell with $0$ -- one with the largest index.

The following lemma establishes \emph{access stability}.
\begin{lemma} \label{lem:access_stability}
Consider an arbitrary network (not necessarily one which sorts all inputs correctly) and an arbitrary cell $c$ in it. Assume that we have access to $c$ on some input $x\in\{0,1\}^n$. Next, assume that we change one input bit of $x$ from $0$ to $1$. Then we can lose the access to $c$ only if this cell gets value 1.
\end{lemma}

\begin{proof}
First, observe that the value of $c$ is a monotone function of the inputs. Indeed, if $c$ is not an input cell, it is an output cell of one of the comparators. Then $c$ is some threshold function of the input cells of this comparator. Namely, it is 1 if and only if sufficiently many inputs to the comparator are 1's. Such functions are monotone. Similarly, these input cells are monotone functions of some cells from the previous array, and so on.

Now, assume that we changed some input bit $x_i$ from 0 to 1. We show that if $c$ is still 0, then we still have access to it. When $x_i$ was $0$, there was some input bit establishing access to $c$. This could not be $x_i$ itself because otherwise $c$ would become $1$. So when $x_i$ was $0$, we could change some other input bit from $0$ to $1$ to make $c$ equal to $1$. By monotonicity of $c$, if we change this input bit from $0$ to $1$ when $x_i$ is already $1$, then $c$ also becomes $1$. So we still have access to $c$ when $x_i$ became equal to $1$.
\end{proof}

Since any output of a sorting network depends on all of its inputs, for any sorting network we have that $k \geqslant n^{1/d}$. Otherwise the fan-in is just not enough to connect every input to an output. We can improve this lower bound almost quadratically using access stability.

\begin{theorem}
For any sorting network with parameters $n$, $k$ and $d$ we have $k \geqslant \left(\frac{n}{2}\right)^{\frac{1}{\lceil d/2 \rceil}}$.
\end{theorem}

\begin{proof}
Assume for contradiction that $k < \left(\frac{n}{2}\right)^{1/\lceil d/2 \rceil}$. Then $k^{\lceil d/2\rceil} < n/2$. On the all-zeros input, we have access to some cell $c$ after the  $\lfloor d/2 \rfloor$th layer. We start switching input cells that are not connected to $c$ from $0$ to $1$, one by one. There are more than $n/2$ of them, because
 $c$ is connected to at most $k^{\lfloor d/2 \rfloor} < n/2$ input cells. Indeed, $k$ is the arity of our network, and there are $\lfloor d/2 \rfloor$ layers of comparators between $c$ and the input cells.

Thus, we will make more than $n/2$ switches. During this, the last array (one with the output cells) is getting filled with ones. It does it in the descending order,  due to the fact that our network is sorting. Namely, the last cell of the array gets 1 first,  then the cell before the last cell, and so on. After the $i$th switch, the cell getting 1 is the $i$th one from the end.

We claim that the last  $n/2$  cells of the last array will be connected to $c$. This will be a contradiction, because $c$ can be connected to at most $k^{\lceil d/2 \rceil} < n/2$ output cells. Indeed, there are $d - \lfloor d/2 \rfloor = \lceil d/2 \rceil$ layers between $c$ and the last array, as $c$ is placed after the  $\lfloor d/2 \rfloor$th layer.

To prove this, we take any $i \le n/2$ and show that the $i$th cell from the end (of the last array) is connected to $c$. Consider the moment before the $i$th switch. Currently, the last $i - 1$ cells of the last array are 1's. We claim that we still have access to $c$. Indeed, we only changed input cells that are not connected with $c$. So $c$ is still $0$, and by Lemma \ref{lem:access_stability} we could not lose access to it. So we can change one of the input bits to make $c$ equal to 1. Then, in the last array, $i$th cell from the end gets 1. However, this cell is a function of the cells after the  $\lfloor d/2 \rfloor$th layer. Among them, only $c$ changed its value from 0 to 1 (we added one 1 to the input, so all the other arrays also got exactly one more 1). So there must be a connection between this cell from the last array and $c$.
\end{proof}

For the rest of our results we will use the following method of lower bounding the fan-in.

\begin{definition}
We call a sequence of $x^1, x^2, \ldots, x^m\in\{0, 1\}^n$ of inputs a \textbf{growing branch} if for every $1 \le i < m$ we have that $x^{i + 1}$ is obtained from $x^i$ by changing one of the input bits from $0$ to $1$.
\end{definition}

\begin{lemma} \label{lem:method}
Consider any sorting network. Assume that there exists a growing branch of length $k - 1$ such that, for every input from the branch, there exist at least 2 cells before the last layer to which we have access on this input.
Then the arity of one of the comparators from the last layer is at least $k$.
\end{lemma}

\begin{proof}
Consider any growing branch $x^1, \ldots, x^{k - 1}\in\{0, 1\}^n$ satisfying assumptions of the lemma. Let $S_i$ be the set of cells before the last layer to which we have access on $x^i$. By assumptions of the lemma, $|S_i| \ge 2$ for every $1 \le i \le k - 1$. Set $S = S_1\cup S_2\cup\ldots \cup S_{k-1}$.
First, we show that $|S| \ge k$. Second, we show that all cells from $S$ go the same comparator in the last layer. So this comparator must have arity at least $k$.

As for the first claim, for every $i < k - 1$ consider a cell before the last layer which becomes $1$ when we switch from $x^i$ to $x^{i+1}$. On one hand, we have access to this cell on $x^i$. On the other hand, we do not have access to it on $x^{i + 1}, \ldots, x^{k - 1}$ (simply because this cell equals 1 on these inputs). Hence, for every $i < k - 1$, we have $|S_i \setminus (S_{i + 1}\cup\ldots \cup S_{k - 1})| \ge 1$. Thus,
\begin{align*}
|S| &\ge  |S_1\setminus (S_2 \cup\ldots\cup S_{k - 1})| + \ldots + |S_{k - 2}\setminus S_{k - 1}| + |S_{k - 1}| \\
&\ge k - 2 + |S_{k - 1}| \ge k.
\end{align*}

As for the second claim, note first that for every $1 \le i \le k - 1$  all cells from $S_i$ go to the same comparator.  Indeed, otherwise we have access to 2 different output cells on $x^i$.

To show that actually all cells from $S$ go to the same comparator in the last layer, we show that $S_i\cap S_{i + 1} \neq \varnothing$ for every $1 \le i < k - 1$. Indeed, $S_i$ contains at least 2 cells. One of them stays $0$ when we switch from $x^i$ to $x^{i + 1}$. By Lemma \ref{lem:access_stability}, we do not lose access to this cell after the switch. Hence, this cell is also in $S_{i + 1}$.
\end{proof}

As a simple illustration of this method, we show that the arity of any sorting network of depth $d \le 2$ is maximal possible.

\begin{proposition}
For any sorting network with parameters $n$, $k$ and $d\leq 2$ we have $k=n$.
\end{proposition}
\begin{proof} For $d=1$, in a network with $k < n$, output cells would be disconnected with some input cells.

Suppose now there is a network with $d=2$ and $k<n$. Call an input to our sorting network \emph{proper} if at least 2 comparators from the first layer receive at least one cell with $0$ on this input. Observe that we have access to at least 2 cells before the last (that is, after the first) layer on any proper input. Indeed, if a comparator from the first layer receives at least one cell with $0$, then we have access to one of the outputs of this comparator. Now, on a proper input we have this for at least 2 different comparators.

To finish the argument, we construct a growing branch of length $n - 1$, consisting of proper inputs. By Lemma \ref{lem:method}, this shows that $k$ must be equal to $n$.

Consider the all-zeros input. It is proper because $k < n$, which means that there are at least 2 comparators in the first layer, and all of them receive only 0's.
Fix two input cells going to different comparators in the first layer. Now we can add 1's to all other cells one by one. Clearly, each input in the resulting growing branch is proper. 

This way we add $1$'s in this manner until their number is $n - 2$. The resulting growing branch of proper inputs will have length $n - 1$.
\end{proof}

\section{Depth 3}
\label{sec:3}
\begin{theorem}
The minimal $k$ for which there exists a sorting network with parameters $n$, $k$ and $d=3$ is $k = \lceil n / 2 \rceil$.
\end{theorem}

The rest of this section is devoted to the proof of this theorem.

\paragraph*{Upper bound}
It is enough to obtain the upper bound for even $n$. Indeed, if $n$ is odd, then just  take a sorting network with parameters $n + 1$, $k = (n + 1)/2 = \lceil n/2\rceil$, $d = 3$ and plug-in 1 into the last input cell. The last cells of all the other arrays will also be 1. So we can just ``disconnect'' all these cells from the comparators they go in. The resulting network correctly sorts all sequences from $\{0, 1\}^n$, and its arity can only be smaller than in the initial network.

From now on, we assume that $n$ is even.
We employ the Batcher's odd-even mergesort~\cite{Knuth98}. Unfortunately, it gives us arity $n/2 + 1$ when $n$ is not divisible by 4. To obtain the optimal arity, we have to modify it a bit. The idea is to merge even elements with odd elements, instead of merging even with even and odd with odd.

In more detail, we represent arrays as tables of the form:
\begin{center}
\begin{tabular}{ c c c c }
 $a_{1,1}$ & $a_{2,1}$ & $a_{3,1}$ & $\ldots$ \\ 
 $a_{1,2}$ & $a_{2,2}$ & $a_{3,2}$ & $\ldots$
\end{tabular},
\end{center}
where the ordering of the cells is as follows:
\[
a_{1,1}, a_{1,2}, a_{2,1}, a_{2,2}, a_{3,1}, a_{3,2},\ldots, a_{n/2, 1}, a_{n/2, 2}.
\]
 In the first layer, we sort both rows. Then we sort the following two sets of cells: $T_0 = \{ a_{i,j} \mid i+j \text{ is even}\}$ and $T_1 = \{ a_{i,j} \mid i+j \text{ is odd}\}$. For example, $T_0$  and $T_1$ look as follows for $n = 10$:
\begin{center}
\begin{tabular}{ c c c c c}
 $a_{1,1}\in T_0$ & $a_{2,1}\in T_1$ & $a_{3,1}\in T_0$ & $a_{4,1}\in T_1$ & $a_{5, 1}\in T_0$\\ 
 $a_{1,2}\in T_1$ & $a_{2,2}\in T_0$ & $a_{3,2}\in T_1$ & $a_{4,2}\in T_0$ & $a_{5, 2}\in T_1$
\end{tabular}.
\end{center}
Finally, in the third layer, we sort all columns, that is, $a_{1, 1}$ and $a_{1,2}$, $a_{2,1}$ and $a_{2,2}$, and so on. Description of the network is finished. We note that the arity of its first two layers is $n/2$, while the arity of the last layers is just $2$.

We now prove that this network is sorting. First, we show that the number of 0's in $T_0$ and the number of 0's in $T_1$ always differ by at most 1 after the first layer. 
For this consider two rows separately. Both rows are sorted after the first layer. In the first row,
 the number of 0's in $T_0$ is either equal to the number of 0's in $T_1$ (if the number of 0's in the row is even), or is greater by 1 (if the number of 0's in the row is odd).
For the second row, 
 the number of 0's in $T_0$ is either equal to the number of 0's in $T_1$ or is less by 1. In total, the number of 0's in $T_0$ and $T_1$ differs by at most 1.

After the second layer,  $T_1$ and $T_0$ are sorted. This means that if we go from left to right, first we see some number of columns with two 0's. Then there might be a column with 0 and 1, but only if the number of 1's in $T_1$ and the number of 1's in $T_0$ differ by 1. After this, there are only columns with two 1's. Therefore, our array is already almost sorted. The only problem is that the column with 0 and 1 might have 1 on the top. But this will be fixed after the third layer.

\paragraph*{Lower bound}
For any sorting network of depth 3 we show the following: if all its comparators from the first and the second layer are of arity less than $\lceil n/2\rceil$, then one of the comparators from the last layer has arity at least\footnote{This claim is quite tight:  we gave a construction in which the arity of the first two layers was  $\lceil n/2\rceil$, while the arity of the third layer was just 2.} $\lfloor n/2\rfloor + 2$. Since $\lfloor n/2\rfloor + 2 > \lceil n/2\rceil$, this establishes our lower bound.

We use Lemma \ref{lem:method}.
First, we consider the all-zeroes input. Then we start adding 1's to it. This gives a growing branch of inputs. We will do this in such a way that for at least $\lfloor n/2\rfloor + 1$ inputs in a row, there will be access to at least 2 cells before the last layer.

To have access to 2 different cells before the last layer (that is, after the second layer), it is sufficient to have access to 2 cells after the first layer that go to different second-layer comparators. Indeed, if we have access to some cell $c$, then we also have access to one of the outputs of the comparator $c$ goes to (namely, to the last $0$-output of this comparator). So we reduced our task to the following. For at least $\lfloor n/2\rfloor + 1$ inputs in a row, we should have access to 2 cells after the first layer that go to different second-layer comparators.

Take any input. Consider cells after the first layer that get value $0$ on this input. In the argument, we represent these cells as \emph{colored cubes} that are arranged in vertical stacks (see Figure~\ref{fig:cubes}). Namely, for each first-layer comparator $C$  there is a stack which consists of cubes corresponding to $0$-outputs of $C$.  In each stack, the order in which cubes go from the bottom to the top is the same as the order of the corresponding cells. In particular, on top of each stack we have the last $0$-output of the comparator in question.

\begin{figure}[h]
\centering

\begin{tikzpicture}[square/.style={regular polygon,regular polygon sides=4},scale=0.8,every node/.style={scale=0.8}]

\draw[->, line width=1mm] (-1,0) -- (-1,7);

	\node at (-2,0) [draw=none] (0) {};
\node at (1.5,7.5) [draw=none] (a) {1st layer};
\node at (4.5,7.5) [draw=none] (b) {2nd layer};
\node at (7.5,7.5) [draw=none] (c) {3rd layer};

        \node at (0,0) [square,draw,thick,minimum width=1.3cm] (l11) {$1$};
        \node at (0,1) [square,draw, thick, minimum width=1.3cm] (l12) {$0$};
        \node at (0,2) [square,draw,thick, minimum width=1.3cm] (l13) {$0$};
        \node at (0,3) [square,draw,thick,minimum width=1.3cm] (l14) {$1$};
        \node at (0,4) [square,draw,thick,minimum width=1.3cm] (l15) {$0$};
        \node at (0,5) [square,draw,thick,minimum width=1.3cm] (l16) {$1$};
        \node at (0,6) [square,draw,thick,minimum width=1.3cm] (l17) {$0$};
        \node at (0,7) [square,draw,thick,minimum width=1.3cm] (l18) {$0$};

		 \node at (3,0) [square,draw,thick,minimum width=1.3cm] (l21) {$0$};
        \node at (3,1) [square,draw, thick, minimum width=1.3cm] (l22) {$0$};
        \node at (3,2) [square,draw,thick, minimum width=1.3cm] (l23) {$1$};
        \node at (3,3) [square,draw,thick,minimum width=1.3cm] (l24) {$1$};
        \node at (3,4) [square,draw,thick,minimum width=1.3cm] (l25) {$0$};
        \node at (3,5) [square,draw,thick,minimum width=1.3cm] (l26) {$0$};
        \node at (3,6) [square,draw,thick,minimum width=1.3cm] (l27) {$0$};
        \node at (3,7) [square,draw,thick,minimum width=1.3cm] (l28) {$1$};

 \draw [-] (l11) -- (l21) node [midway, draw,circle,fill=one] {};
 \draw [-] (l12) -- (l22) node [midway, draw,circle,fill=one] {};
 \draw [-] (l13) -- (l23) node [midway, draw,circle,fill=one] {};
 \draw [-] (l14) -- (l24) node [midway, draw,circle,fill=one] {};
 \draw [-] (l15) -- (l25) node [midway, draw,circle,fill=two] {};
 \draw [-] (l16) -- (l26) node [midway, draw,circle,fill=two] {};
 \draw [-] (l17) -- (l27) node [midway, draw,circle,fill=two] {};
 \draw [-] (l18) -- (l28) node [midway, draw,circle,fill=two] {};

		 \node at (6,0) [square,draw,thick,minimum width=1.3cm] (l31) {$0$};
        \node at (6,1) [square,draw, thick, minimum width=1.3cm] (l32) {$0$};
        \node at (6,2) [square,draw,thick, minimum width=1.3cm] (l33) {$0$};
        \node at (6,3) [square,draw,thick,minimum width=1.3cm] (l34) {$1$};
        \node at (6,4) [square,draw,thick,minimum width=1.3cm] (l35) {$0$};
        \node at (6,5) [square,draw,thick,minimum width=1.3cm] (l36) {$0$};
        \node at (6,6) [square,draw,thick,minimum width=1.3cm] (l37) {$1$};
        \node at (6,7) [square,draw,thick,minimum width=1.3cm] (l38) {$1$};

 \draw [-] (l21) -- (l31) node [midway, draw,circle,fill=one] {};
 \draw [-] (l22) -- (l32) node [midway, draw,circle,fill=one] {};
 \draw [-] (l23) -- (l33) node [midway, draw,circle,fill=two] {};
 \draw [-] (l24) -- (l34) node [midway, draw,circle,fill=two] {};
 \draw [-] (l25) -- (l35) node [midway, draw,circle,fill=one] {};
 \draw [-] (l26) -- (l36) node [midway, draw,circle,fill=one] {};
 \draw [-] (l27) -- (l37) node [midway, draw,circle,fill=two] {};
 \draw [-] (l28) -- (l38) node [midway, draw,circle,fill=two] {};

\node at (9,0) [square,draw,thick,minimum width=1.3cm] (l41) {$0$};
        \node at (9,1) [square,draw, thick, minimum width=1.3cm] (l42) {$0$};
        \node at (9,2) [square,draw,thick, minimum width=1.3cm] (l43) {$0$};
        \node at (9,3) [square,draw,thick,minimum width=1.3cm] (l44) {$0$};
        \node at (9,4) [square,draw,thick,minimum width=1.3cm] (l45) {$0$};
        \node at (9,5) [square,draw,thick,minimum width=1.3cm] (l46) {$1$};
        \node at (9,6) [square,draw,thick,minimum width=1.3cm] (l47) {$1$};
        \node at (9,7) [square,draw,thick,minimum width=1.3cm] (l48) {$1$};

 \draw [-] (l31) -- (l41) node [midway, draw,circle,fill=one] {};
 \draw [-] (l32) -- (l42) node [midway, draw,circle,fill=one] {};
 \draw [-] (l33) -- (l43) node [midway, draw,circle,fill=two] {};
 \draw [-] (l34) -- (l44) node [midway, draw,circle,fill=two] {};
 \draw [-] (l35) -- (l45) node [midway, draw,circle,fill=two] {};
 \draw [-] (l36) -- (l46) node [midway, draw,circle,fill=two] {};
 \draw [-] (l37) -- (l47) node [midway, draw,circle,fill=one] {};
 \draw [-] (l38) -- (l48) node [midway, draw,circle,fill=one] {};

            \cube[fill=one]{4,-5,0}{1}
	\cube[fill=one]{4,-4,0}{1}

\cube[fill=one]{6,-5,0}{1}
	\cube[fill=one]{6,-4,0}{1}
	\cube[fill=two]{6,-3,0}{1}
   
    \end{tikzpicture}

\caption{Stacks of cubes. The first stack comes from the dark gray comparator in the first layer. It has two 0-outputs that both go to dark gray comparator in the second layer. So both cubes in the first stack are dark gray. The second stack comes from the light gray comparator in the first layer. One of its 0-outputs (one with the largest index) goes to the light gray comparator in the second layer. So the top cube in the second stack is light gray.}
\label{fig:cubes}
\end{figure}

To obtain a coloring of these cubes, we first assign a distinct color to each  comparator from the second layer. Then, given a cube, we consider the corresponding cell, look at the comparator this cell goes to in the second layer, and color our cube into the color of this comparator.

We have access to a cell after the first layer if and only if it is the last 0-output of some first-layer comparator. Such cells are represented by the top cubes of our stacks. Recall that we want to have access to 2 cells after the first layer that go to different second-layer
comparators. This means to have \emph{2 top cubes of different colors}.

Now, adding 1's to our input is equivalent to removing the top cubes one by one in some order. Indeed, assume we change one of the input cells from $0$ to $1$. Consider a first-layer comparator $C$ this cell goes to. The last $0$-output of $C$ now gets $1$. All the other cells after the first layer do not change. This just means that the top cube of the stack assigned to $C$ gets removed.

Since the arity of any first-layer comparator is less than $\lceil n/2\rceil$, each vertical stack has less than $\lceil n/2\rceil$ cubes. Similarly, since the arity of any second-layer comparator is less than $\lceil n/2\rceil$, for each color there are less than $\lceil n/2\rceil$ cubes of this color. Thus, our task reduces to the following \emph{Cubes problem}.

\medskip

\textbf{Cubes problem.} There are $n$ cubes, arranged in several vertical stacks. Each cube has a color. Each stack has less than $\lceil n/2\rceil$ cubes (or equivalently, less than half of the cubes). For each color, there are less than $\lceil n/2\rceil$ cubes (or equivalently, less than half of the cubes) of this color.

In one step, we can remove a cube from the top of one of the stacks. Show that there exists a way of removing cubes such that at least $\lfloor n/2\rfloor + 1$ times in a row, not necessarily right from the start, there exist 2 top cubes of different colors (we stress that at different moments this might be different cubes).

\medskip

\textbf{Solution of the problem.}
 We call a stack \emph{monochromatic} if all its cubes are of the same color. 
For notation convenience, we set $m = \lceil n/2\rceil$.
 Our argument can be split into the following two lemmas.
\begin{lemma}
\label{lem:first_step}
There is a way of removing at most $\lfloor m/2\rfloor - 1$ cubes from the initial configuration of cubes such that as the result we either have that
\begin{equation}\label{eqn:A}
   \mbox{there exist 2 monochromatic stacks of different colors;}\tag{A}
  \end{equation}
or that
\begin{align}\label{eqn:B}
\tag{B}
\begin{cases}
&\mbox{there exist 2 top cubes of different colors,}\\
&\mbox{there are at least 3 stacks},\\ 
&\mbox{and at least 2 stacks are non-monochrormatic.}
\end{cases}
\end{align}
\end{lemma}

\begin{lemma} 
\label{lem:second_step}
Consider any configuration of cubes which satisfies either \eqref{eqn:A} or \eqref{eqn:B}. There exists an algorithm of removing cubes from this configuration for which the following holds: \textbf{(a)} when the algorithm stops, at most $\lfloor m/2\rfloor + 1$ cubes are left; \textbf{(b)} throughout the execution of the algorithm (in particular, when the algorithm stops), there exist 2 top cubes of different colors.
\end{lemma}

Let us explain why these two lemmas imply a solution to our problem. First, using Lemma \ref{lem:first_step}, we remove  $s\le \lfloor m/2\rfloor - 1$ cubes so that now our configuration satisfies  either \eqref{eqn:A} or \eqref{eqn:B}. There are $n - s$ cubes left. Then we apply the algorithm of Lemma \ref{lem:second_step}. When it finishes, there are $r \le \lfloor m/2 \rfloor + 1$ cubes left.  So $t = n - s - r + 1$ times in a row we had 2 top cubes of different colors. It remains to show that $t\ge  \lfloor n/2\rfloor + 1$. Indeed, 
\begin{align*}
t &\ge n -  (\lfloor m/2\rfloor - 1) -   (\lfloor m/2\rfloor + 1) + 1 \\
&= n + 1 - 2\cdot\lfloor m/2\rfloor \ge n + 1 - m = n + 1 - \lceil n/2 \rceil  \\
&= \lfloor n/2\rfloor + 1.
\end{align*}

\begin{proof}[Proof of Lemma \ref{lem:first_step}]
If the initial configuration already satisfies \eqref{eqn:A}, there is nothing left to do. Assume from now on that it does not. Then all the monochromatic stacks are of the same color. This implies that there are at least 2 non-monochromatic stacks. Indeed, otherwise there is at most one non-monochromatic stack, it has less than half of the cubes, so the color of the monochromatic stacks occupies more than half of the cubes.

Note, that there are at least 3 different stacks because each stack has less than half of the cubes

Now, if initially there exist 2 top cubes of different colors, then the initial configuration already satisfies \eqref{eqn:B}. Assume from now one that initially all top cubes have the same color. Take any 2 non-monochromatic stacks. Both of them contain a cube whose color differs from the top cubes. In both stacks, mark the highest cube with this property (so that all cubes above them have the same color as all the top cubes). Among the two marked cubes select one which has the least distance to the top (that is, one which has the least number of cubes above it in its stack). Remove all the cubes above the selected one. Now it is on the top. It remains to explain two things: why do we remove at most $\lfloor m/2\rfloor - 1$ cubes, and why does the resulting configuration satisfy either \eqref{eqn:A} or \eqref{eqn:B}.

\textbf{Why do we remove at most $\lfloor m/2\rfloor - 1$ cubes.} Consider cubes above the marked ones. It is sufficient to show that there are at most $m - 2$ of them (then one of the marked cubes has at most $\lfloor(m - 2)/2\rfloor = \lfloor m/2\rfloor - 1$ cubes above it). Assume for contradiction that there are at least $m - 1$ of them. As we pointed out, all of them have the same color as all the top cubes.
Besides the two stacks under consideration, there exists at least one more stack, and its top cube also has the same color. So we already have at least $(m - 1) + 1 = \lceil n/2\rceil$ cubes of the same color, which is impossible.

\textbf{Why does the resulting configuration satisfy either \eqref{eqn:A} or \eqref{eqn:B}}. Due to the selected cube,  now we have cubes of different colors on the top.
There are still at least 3 non-empty stacks, because we did not make any of the stacks empty. If there are at least 2 non-monochromatic stacks, then we have \eqref{eqn:B}. Assume now that there is at most 1 non-monochromatic stack.  Initially, we had at least 2 non-monochromatic stacks. Hence, the stack with the selected cube became monochromatic (other stacks did not change). There are at least 3 stacks, at most 1 is non-monochromatic, so besides the stack with the selected cube there is at least one more monochromatic stack $S$. The top cube of $S$ was on the top initially, therefore its color differs from the color of the selected cube. So $S$ and the stack with the selected cube are two monochromatic stacks of different colors. Hence, we have \eqref{eqn:A}.
\end{proof}

\begin{proof}[Proof of Lemma \ref{lem:second_step}]
We call a configuration of cubes \emph{terminal} if one of the following conditions holds:
\begin{itemize}
\item there are at most 2 cubes;
\item there are exactly 3 non-empty stacks, one with exactly one cube (let the color of this cube be $c$), and each of the other 2 stacks is as follows: it has at least 2 cubes, the color of the top cube is different from $c$, and all the other cubes of this stack are colored in $c$.
\end{itemize}

We show that it is possible to reach a terminal configuration while maintaining a property that there exist 2 top cubes of different colors. Let us explain why is this sufficient to establish the lemma. If there are at most 2 cubes, then, since $\lceil m/2\rceil + 1 
\ge 2$, there is nothing left to do\footnote{Here we assume that $n\geq 2$, but sorting networks for $n=1$ do not make much sense.}. Assume now that the second condition from the definition of a terminal configuration holds. There are at most $m - 1$ cubes colored in $c$. Hence, one of the non-monochromatic stacks has at most $\lfloor (m - 2)/2 \rfloor = \lfloor m/2\rfloor - 1$ cubes colored in $c$. On the other hand, all its non-top cubes are colored in $c$, so its size is at most $\lfloor m/2\rfloor$. This non-monochromatic stack and  the stack of size 1 provide 2 top cubes of different colors, regardless of what happens with the third stack. So if remove all cubes from the third stack, we will be left with at most $\lfloor m/2\rfloor + 1$ cubes, as required.

Now we show how to reach a terminal configuration. Our algorithm repeatedly does the following. In any configuration, it first tries to remove a cube in such a way that the resulting configuration satisfies \eqref{eqn:A}. If this is impossible, the algorithm tries to remove a cube in such a way that the resulting configuration satisfies \eqref{eqn:B}. If this is impossible as well, the algorithm ``gets stuck''.

The algorithm maintains the OR of \eqref{eqn:A} and \eqref{eqn:B}. So it also maintains the existence of 2 top cubes of different colors, as both \eqref{eqn:A} and \eqref{eqn:B} imply this.
We now show that if $C$ is a configuration after the algorithm got stuck, then $C$ is terminal. We know that $C$ satisfies either \eqref{eqn:A} or \eqref{eqn:B}, but no cube can be removed from it in such a way that \eqref{eqn:A} or \eqref{eqn:B} still holds.

If $C$ satisfies \eqref{eqn:A}, then there are at most 2 cubes. Indeed, consider two monochromatic stacks of different colors. If one of these stacks has at least 2 cubes, then we can take the top cube from it, and we will still have \eqref{eqn:A}, which is impossible. So both stacks are of size 1. If there were other stacks, we could take something from them, and we would still have \eqref{eqn:A}. This shows that there are just 2 cubes.

Assume now that $C$ satisfies \eqref{eqn:B} but not \eqref{eqn:A}. We show that $C$ satisfies the second condition of the definition of a terminal configuration.

Due to \eqref{eqn:B}, there are at least 2 non-monochromatic stacks in $C$. We claim that there are exactly 2 non-monochromatic stacks. Indeed, assume for contradiction that there are at least 3 non-monochromatic stacks. It is possible to remove the top cube from one of them so that on the top we still have 2 different colors. We will also have at least 2 non-monochromatic stacks and at least 3 stacks in total (one cannot make a non-monochromatic stack empty in one step).
So will still have \eqref{eqn:B}, which is impossible.

So in $C$ there are exactly 2 non-monochromatic stacks. As in $C$ there are at least 3 stacks, $C$ must have monochromatic stacks. Since $C$ does not satisfy \eqref{eqn:A}, all monochromatic stacks of $C$ are of the same colors. We claim that there is exactly 1 monochromatic stack, and it has exactly 1 cube (let $c$ denote its color). Indeed, assume that this is not the case. Then take something from the monochromatic stacks. There will still be 2 non--monochromatic stacks and something else. The set of colors on the top will not change  because all monochromatic stacks are of the same color. So our configuration will still satisfy \eqref{eqn:B}, which is impossible.

To summarize, $C$ consists of  1 separate cube of color $c$ and 2 non-monochromatic stacks. Denote these two stacks by $S_1$ and $S_2$. Since $C$ satisfies \eqref{eqn:B}, on the top we have 2 cubes of different colors. W.l.o.g., the top cube of $S_1$ is colored not in $c$. Remove the top cube of $S_2$. We still have 2 top cubes of different colors. Since the resulting configuration does not satisfy \eqref{eqn:B}, the stack $S_2$ must have become monochromatic. Moreover, its  must be now colored in $c$, because otherwise the resulting configuration satisfies \eqref{eqn:A}. In other words, all the cubes of $S_2$, except the top one, must be colored in $c$. The top cube of $S_2$ must be colored differently from $c$, because $S_2$ is non-monochromatic. So $S_2$ is as in the definition of a terminal configuration.

Recall that we have shown this for $S_2$ assuming that the color of the top cube of $S_1$ differs from $c$. We have established that the color of the top cube of $S_2$ differs from $c$. Hence, we can now apply a similar argument to $S_1$ to show that this stack is also as in the definition of a terminal configuration.
\end{proof}

\section{Depth 4}
\label{sec:4}
\begin{theorem}
The minimal $k$ for which there exists a sorting network with parameters $n$, $k$ and $d=4$ is $k = \Theta(n^{2/3})$.
\end{theorem}
The rest of the section is devoted to the proof of this theorem.

\paragraph*{Upper bound}
The upper bound is implicit in~\cite{Leighton85}. Leighton there gives a sorting algorithm called ColumnSort. It can be converted into a sorting network of depth $d = 4$ and arity $k = \Theta(n^{2/3})$. Namely, the input array in ColumnSort is represented as an $O(n^{2/3})\times O(n^{1/3})$ matrix. There are 4 steps in the algorithm when it sorts each column of the matrix. Each of these steps can be seen as a layer of comparators of arity $O(n^{2/3})$ (this is how many entries are in each column). Between these 4 steps, ColumnSort performs some permutations of the entries of the matrices, but these permutations do not depend on the values of the entries. So, these permutations just define how entries are partitioned between the comparators.

\paragraph*{Lower bound}
Suppose there is a depth-$4$ sorting network such that all its comparators from the first 3 layers have arity at most $t = \frac{n^{2/3}}{100}$. We show that one of the comparators from the fourth layer has arity $\Omega(n^{2/3})$. As in the previous proof, we start with the all-zeros input, and then we add 1's to our input one by one. By Lemma~\ref{lem:method}, it is enough to have access to at least 2 cells before the last layer (or, equivalently, after the third layer) for at least $\Omega(n^{2/3})$ consecutive steps. In turn, to have access to at least 2 different cells after the third layer, it is sufficient to have access to 2 cells after the second layer that go to different third-layer comparators.

First, it is convenient for us to assume that all first-layer comparators have arity at least $t$ and at most $3t$. We can achieve this by joining some first-layer comparators. Namely, while there are at least 2 comparators with arity less than $t$, we join them into one comparator of arity less than $2t$. If there is already just 1 comparator with arity less than $t$, we can add it to any comparator of arity less than $2t$.

The resulting network still sorts all inputs correctly. This is because when we join two first-layer comparators, the set of arrays that can appear after the first layer only becomes smaller.

The rest of the argument is again presented in terms of cubes, see Figure \ref{fig:cubes2}. Fix an arbitrary input to our network.
  We will have two sets of cubes assigned to this input. Cubes from the first set will be identified with cells after the first layer that get value $0$ on this input. As before, we split these cubes into vertical stacks. Each stack corresponds to one of the comparators of the first layer. More specifically, for each comparator of the first layer, we stack its $0$-outputs upon each other as cubes. The order in which they go from the bottom to the top coincides with their order in the network, and on the top we have a cell with the largest index. Additionally, we label these cubes by numbers from $1$ to $a$, where $a$ is the number of second-layer comparators. Namely, each of these cubes is assigned the index of the comparator this cube (or rather the corresponding cell) goes to in the second layer.

Now, cubes from the second set will be identified with cells after the second layer that get value $0$ on our input. We arrange them into vertical stacks according to second-layer comparators, similarly to cubes from the first set. Additionally, cubes from the second set will be colored. Namely, we first assign a distinct color to each of the comparators from the third layer. Then every cube from the second set is colored into the color of the comparator this cube (or rather the corresponding cell) goes to in the third layer.

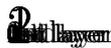
\begin{figure}[h]
\centering

\begin{tikzpicture}[square/.style={regular polygon,regular polygon sides=4},scale=0.8,every node/.style={scale=0.8}]

\draw[->, line width=1mm] (-1,0) -- (-1,7);

	\node at (-2,0) [draw=none] (0) {};
\node at (1.5,7.5) [draw=none] (a) {1st layer};
\node at (4.5,7.5) [draw=none] (b) {2nd layer};
\node at (7.5,7.5) [draw=none] (c) {3rd layer};
\node at (10.5,7.5) [draw=none] (d) {4th layer};

        \node at (0,0) [square,draw,thick,minimum width=1.3cm] (l11) {$1$};
        \node at (0,1) [square,draw, thick, minimum width=1.3cm] (l12) {$0$};
        \node at (0,2) [square,draw,thick, minimum width=1.3cm] (l13) {$0$};
        \node at (0,3) [square,draw,thick,minimum width=1.3cm] (l14) {$1$};
        \node at (0,4) [square,draw,thick,minimum width=1.3cm] (l15) {$0$};
        \node at (0,5) [square,draw,thick,minimum width=1.3cm] (l16) {$1$};
        \node at (0,6) [square,draw,thick,minimum width=1.3cm] (l17) {$0$};
        \node at (0,7) [square,draw,thick,minimum width=1.3cm] (l18) {$0$};

		 \node at (3,0) [square,draw,thick,minimum width=1.3cm] (l21) {$0$};
        \node at (3,1) [square,draw, thick, minimum width=1.3cm] (l22) {$0$};
        \node at (3,2) [square,draw,thick, minimum width=1.3cm] (l23) {$1$};
        \node at (3,3) [square,draw,thick,minimum width=1.3cm] (l24) {$1$};
        \node at (3,4) [square,draw,thick,minimum width=1.3cm] (l25) {$0$};
        \node at (3,5) [square,draw,thick,minimum width=1.3cm] (l26) {$0$};
        \node at (3,6) [square,draw,thick,minimum width=1.3cm] (l27) {$0$};
        \node at (3,7) [square,draw,thick,minimum width=1.3cm] (l28) {$1$};

 \draw [-] (l11) -- (l21) node [midway, draw,circle,fill=one] {};
 \draw [-] (l12) -- (l22) node [midway, draw,circle,fill=one] {};
 \draw [-] (l13) -- (l23) node [midway, draw,circle,fill=one] {};
 \draw [-] (l14) -- (l24) node [midway, draw,circle,fill=one] {};
 \draw [-] (l15) -- (l25) node [midway, draw,circle,fill=two] {};
 \draw [-] (l16) -- (l26) node [midway, draw,circle,fill=two] {};
 \draw [-] (l17) -- (l27) node [midway, draw,circle,fill=two] {};
 \draw [-] (l18) -- (l28) node [midway, draw,circle,fill=two] {};

		 \node at (6,0) [square,draw,thick,minimum width=1.3cm] (l31) {$0$};
        \node at (6,1) [square,draw, thick, minimum width=1.3cm] (l32) {$0$};
        \node at (6,2) [square,draw,thick, minimum width=1.3cm] (l33) {$0$};
        \node at (6,3) [square,draw,thick,minimum width=1.3cm] (l34) {$1$};
        \node at (6,4) [square,draw,thick,minimum width=1.3cm] (l35) {$0$};
        \node at (6,5) [square,draw,thick,minimum width=1.3cm] (l36) {$1$};
        \node at (6,6) [square,draw,thick,minimum width=1.3cm] (l37) {$0$};
        \node at (6,7) [square,draw,thick,minimum width=1.3cm] (l38) {$1$};

 \draw [-] (l21) -- (l31) node [midway, draw,circle,fill=white] {1};
 \draw [-] (l22) -- (l32) node [midway, draw,circle,fill=white] {2};
 \draw [-] (l23) -- (l33) node [midway, draw,circle,fill=white] {2};
 \draw [-] (l24) -- (l34) node [midway, draw,circle,fill=white] {2};
 \draw [-] (l25) -- (l35) node [midway, draw,circle,fill=white] {1};
 \draw [-] (l26) -- (l36) node [midway, draw,circle,fill=white] {2};
 \draw [-] (l27) -- (l37) node [midway, draw,circle,fill=white] {1};
 \draw [-] (l28) -- (l38) node [midway, draw,circle,fill=white] {1};

\node at (9,0) [square,draw,thick,minimum width=1.3cm] (l41) {$0$};
        \node at (9,1) [square,draw, thick, minimum width=1.3cm] (l42) {$0$};
        \node at (9,2) [square,draw,thick, minimum width=1.3cm] (l43) {$0$};
        \node at (9,3) [square,draw,thick,minimum width=1.3cm] (l44) {$0$};
        \node at (9,4) [square,draw,thick,minimum width=1.3cm] (l45) {$1$};
        \node at (9,5) [square,draw,thick,minimum width=1.3cm] (l46) {$1$};
        \node at (9,6) [square,draw,thick,minimum width=1.3cm] (l47) {$0$};
        \node at (9,7) [square,draw,thick,minimum width=1.3cm] (l48) {$1$};

 \draw [-] (l31) -- (l41) node [midway, draw,circle,fill=one] {};
 \draw [-] (l32) -- (l42) node [midway, draw,circle,fill=one] {};
 \draw [-] (l33) -- (l43) node [midway, draw,circle,fill=two] {};
 \draw [-] (l34) -- (l44) node [midway, draw,circle,fill=two] {};
 \draw [-] (l35) -- (l45) node [midway, draw,circle,fill=two] {};
 \draw [-] (l36) -- (l46) node [midway, draw,circle,fill=two] {};
 \draw [-] (l37) -- (l47) node [midway, draw,circle,fill=one] {};
 \draw [-] (l38) -- (l48) node [midway, draw,circle,fill=one] {};

\node at (12,0) [square,draw,thick,minimum width=1.3cm] (l51) {$0$};
        \node at (12,1) [square,draw, thick, minimum width=1.3cm] (l52) {$0$};
        \node at (12,2) [square,draw,thick, minimum width=1.3cm] (l53) {$0$};
        \node at (12,3) [square,draw,thick,minimum width=1.3cm] (l54) {$0$};
        \node at (12,4) [square,draw,thick,minimum width=1.3cm] (l55) {$0$};
        \node at (12,5) [square,draw,thick,minimum width=1.3cm] (l56) {$1$};
        \node at (12,6) [square,draw,thick,minimum width=1.3cm] (l57) {$1$};
        \node at (12,7) [square,draw,thick,minimum width=1.3cm] (l58) {$1$};

           \draw [-] (l41) -- (l51) node [midway, draw,circle,fill=one] {};
 \draw [-] (l42) -- (l52) node [midway, draw,circle,fill=one] {};
 \draw [-] (l43) -- (l53) node [midway, draw,circle,fill=one] {};
 \draw [-] (l44) -- (l54) node [midway, draw,circle,fill=one] {};
 \draw [-] (l45) -- (l55) node [midway, draw,circle,fill=two] {};
 \draw [-] (l46) -- (l56) node [midway, draw,circle,fill=two] {};
 \draw [-] (l47) -- (l57) node [midway, draw,circle,fill=two] {};
 \draw [-] (l48) -- (l58) node [midway, draw,circle,fill=two] {};

 \cube[fill=white]{0,-5,0}{1}
	\cube[fill=white]{0,-4,0}{1}

\cube[fill=white]{2,-5,0}{1}
	\cube[fill=white]{2,-4,0}{1}
	\cube[fill=white]{2,-3,0}{1}
   
\node at (-0.5,-4.5) [draw=none] (f1) {\huge 1};
\node at (-0.5,-3.5) [draw=none] (f2) {\huge 2};

\node at (1.5,-4.5) [draw=none] (f3) {\huge 1};
\node at (1.5,-3.5) [draw=none] (f4) {\huge 2};
\node at (1.5,-2.5) [draw=none] (f5) {\huge 1};

\node at (8,-1.3) [draw=none] (f5) {\Large 1st};

 \cube[fill=one]{8,-5,0}{1}
	\cube[fill=two]{8,-4,0}{1}
	\cube[fill=one]{8,-3,0}{1}
\node at (10,-2.3) [draw=none] (f5) {\Large 2nd};

\cube[fill=one]{10,-5,0}{1}
	\cube[fill=two]{10,-4,0}{1}

    \end{tikzpicture}

\caption{Second-layer comparators are given by numbers 1 and 2, whereas other comparators are given by colors. The first set of cubes is on the left. For example, the first stack there comes from the dark gray first-layer comparator. It has two 0-outputs, one with a smaller index goes to the 1st second-layer comparator, and one with a larger index goes to the 2nd second-layer comparator. So, the cubes are labeled by 1 and 2.
The second set of cubes is on the right. They are colored according to the same principles as in the case of depth 3, but one layer to the right. Additionally, next to each stack we write the index of the corresponding second-layer comparator.
}
\label{fig:cubes2}
\end{figure}

Observe that the number of cubes from the first set that are labeled by $i$ coincides with the size of the $i$th stack from the second set. Indeed, the first number is the number of $0$-inputs to the $i$th second-layer comparator, and the second number is the number of $0$-outputs of this comparator. Additionally, the size of any first-set stack is from $t$ to $3t$, and the size of any second-set stack is at most $t$, because of the arity of the comparators from the first two layers. Finally, for any color, there are at most $t$ cubes of this color because any comparator from the third layer has arity at most $t$.

As in the previous proof, adding 1's to an input is equivalent to removing top cubes from the first set. Now, what happens with cubes from the second set in this process? Assume that we remove some top cube from the first set. Let its label be $i$. This means that some cell after the first layer which goes to $i$th second-layer comparator switches from $0$ to $1$. As a result, the last $0$-output of the $i$th second-layer comparator also switches from $0$ to $1$. This means that the top cube of the $i$th second-set stack gets removed.

In other words, if we have a cube from the first set which is on the top and whose label is $i$, then we also have access to the top cube of the $i$th second-set stack. Now, our goal is to maintain access to 2 second-set cubes of different colors for $\Omega(n^{2/3})$ steps (because having different colors means going into different third-layer comparators). That is, our task reduces to the following problem.

\medskip

\textbf{The second cubes problem.}
There are $n$ \emph{left cubes} and $n$ \emph{right\footnote{We will use an intuition from Figure \ref{fig:cubes2}, where numbered cubes are on the left and colored cubes are on the right.} cubes}. Left cubes are labeled by numbers from $\{1, 2, \ldots, a\}$. Right cubes are colored. All $2n$ cubes are arranged in vertical stacks. We do not mix left and right cubes. That is, there are two types of stacks: those that consist of left cubes (we call them left stacks) and those that consist of right cubes (we call them right stacks). There are exactly $a$ right stacks, they are indexed by numbers from $1$ to $a$.

There are the following restrictions. First,  for every $1 \le i \le a$, the size of the $i$th right stack must be equal to the number of left cubes with label $i$. Next, define $t = \frac{n^{2/3}}{100}$. Each left stack must have size from $t$ to $3t$, and each right stack must have size at most $t$. Finally, for each color, there must be at most $t$ right cubes that have this color.

At one step, we can remove the top cube of one of the left stacks. If we remove a cube with label  $i\in\{1, \ldots, a\}$, then we are also obliged to remove the top cube from the $i$th right stack (and we cannot remove any other cubes in this step).

If there is a top left cube whose label is $i$, we say that it \emph{gives access} to the color of the top cube of the $i$th right stack. \emph{Show} that there is a way of removing cubes such that $\Omega(n^{2/3})$ times in a row (not necessarily right from the start) we have access to at least 2 different colors.

\medskip

\textbf{Solution of the problem.} 

Let $m$ be the number of left stacks. Since each left stack has from $t$ to $3t$ cubes, we have that $m = \Theta(n/t) = \Theta(n^{1/3})$. Set $l = \lfloor n^{1/3}\rfloor$. In each left stack, choose a random consecutive substack of length $l$ (uniformly and independently for each stack).

\begin{lemma}
\label{prob_lemma} With positive probability, for every $i\in\{1, \ldots, a\}$, there will be less than $m/10$ substacks with $i$-labeled cubes. 
\end{lemma}
\begin{proof}
Fix $i\in\{1, \ldots, a\}$. Let $a_j$ be the number of $i$-labeled cubes in the $j$th left stack. Note that $a_1 + \ldots + a_m$ is the size of the $i$th right stack. Hence, this sum does not exceed $t$. 

Let $p_j$ for $j=1,\ldots,m$ be the probability that the $j$th substack contains an $i$-labeled cube. We claim that $p_j \le \frac{2l \cdot a_j}{t}$.  This is because $\frac{2l \cdot a_j}{t}$ is an upper bound on the expected number of $i$-labeled cubes in the $j$th substack. Indeed, since the size of any left stack is at least $t$, the number of $l$-length substacks in each of them is at least $t - l + 1$. On the other hand, each cube is covered by at most $l$ substacks. So, an $i$-labeled cube from the $j$th stack falls into a random $l$-length substack with probability at most $\frac{l}{t - l + 1}$. In turn, $\frac{l}{t - l + 1} \le \frac{2l}{t}$ because  $l = \Theta(n^{1/3})$ and $t = \Theta(n^{2/3})$. Multiplying $\frac{2l}{t}$ by the number of $i$-labeled cubes in the $j$th stack, we obtain an upper bound on the expected number of $i$-labeled cubes in the $j$th substack.

So, the expected number of substacks with $i$-labeled cubes is $p_1 + \ldots + p_m \le \frac{2l \cdot a_1}{t} + \ldots + \frac{2l \cdot a_m}{t} \le 2l$. By Hoeffding's inequality~\cite[Theorem 2]{hoeffding1963probability}, the probability that this number exceeds $3 l$ is 
\[\exp\left\{-\Omega\left(m\cdot\frac{l^2}{m^2}\right)\right\} =\exp\left\{-\Omega\left(\frac{l^2}{m}\right)\right\} = \exp\{-\Omega(n^{1/3})\}\] 
(recall that $m =\Theta(n^{1/3})$ and $l = \Theta(n^{1/3})$). Therefore, by the union bound, with positive probability the number of substacks with $i$-labeled cubes is at most $3l$ for every $i\in\{1,\ldots, a\}$. It remains to notice that $3 l$ is smaller than $m/10$. This is because there are $n$ left cubes, and each left stack has at most $3t$ cubes, so the number of left stacks is $m\ge n/(3t) =  \frac{100 \cdot n^{1/3}}{3} = \frac{100}{9} \cdot 3 n^{1/3} \implies 3l \le 3 n^{1/3} \le  \frac{9}{100}\cdot m < m/10$.
\end{proof}
We fix any choice of substacks such that there are less than $m/10$ substacks with $i$-labeled cubes, for every $i\in\{1, \ldots, a\}$. First, remove all cubes above these substacks. Then we will start taking cubes from the substacks. It is prohibited to take the last cube from any of the substacks. We will call a substack with only 1 cube left \emph{exhausted}. 

Let $S$ denote the set of colors to which we currently have  access. Our goal is to have $|S|\ge2$ for $\Omega(n^{2/3})$ times in a row. We achieve this as follows. If $|S| \ge 3$, we take an arbitrary cube, unless all substacks are exhausted. Note that this can decrease $|S|$ by at most 1. Indeed, the cube that we have removed gave access to just one color. Now, if $|S| \le 2$, consider any $|S|$ left cubes that establish access to $S$. Let their labels be $i$ and $j$ (if $|S| = 1$, we have $i = j$). If possible, remove any left cube which belongs to a non-exhausted substack and whose label differs from $i$ and $j$. Observe that we can only add a new color to $S$ because we still have access to old colors through the same cube, and nothing was taken from the $i$th and from the $j$th right stacks. In particular, when $|S|\le 2$, the size of $S$ cannot decrease.

Note that unless $m/2$ substacks are already exhausted, it is still possible to take one more cube according to the above rules. Indeed, if $|S| \ge 3$, then we need just 1 non-exhausted substack. Now, if $|S| \le 2$, then we can take one more cube unless all top cubes of non-exhausted substacks have labels $i$ or $j$. This can happen only if either at least $m/4$ substacks have $i$-labeled cubes or at least $m/4$ substacks have $j$-labeled cubes. But this is impossible due to our choice of substacks.

These considerations show that we get stuck only if we have already taken at least $m/2 \cdot (l - 1)$ cubes. Another observation is that once the size of $S$ became larger than $1$, it can never become 1 again. So we can have $|S| = 1$ for some number of steps in the beginning, and then we will always have $|S| \ge 2$. Finally, note that we can have $|S| = 1$ for at most $t$ steps. Indeed, throughout the whole period when $|S| = 1$, we have access to the same color $c$ (recall that when $|S|\le 2$, colors cannot be removed from $S$, only a new color can be added). Each time we make a step, we remove some right cube whose color is $c$, otherwise we would have had access to another color. But there are at most $t$ cubes whose color is $c$.

Thus, we will have $|S| \ge 2$ for at least $m/2 \cdot (l - 1) - t$ steps. Since $m \ge n/(3t)$, $l = \lfloor n^{1/3}\rfloor$ and $t = \frac{n^{2/3}}{100}$, the quantity $m/2 \cdot (l - 1) - t$ is $\Omega(n^{2/3})$.

\end{document}